\documentclass[11pt,a4paper]{article}
\usepackage[margin=2.5cm]{geometry}
\usepackage[T1]{fontenc}

\title{Unit Interval Selection in Random Order Streams}

\author{Cezar-Mihail Alexandru\thanks{\texttt{sorin\_cezar@yahoo.com}, Independent Researcher, Bristol, UK}
        \and Adithya Diddapur\footnotemark[3]
        \and Magn\'us M. Halld\'orsson\thanks{\texttt{mmh@ru.is}, ICE-TCS \& Department of Computer Science, Reykjavik University, Iceland. Partially supported by Icelandic Research Fund grant 2511609.}
        \and Christian Konrad\thanks{\texttt{\{adi.diddapur,christian.konrad\}@bristol.ac.uk}, School of Computer Science, University of Bristol, Bristol, UK.}
        \and Kheeran K. Naidu\thanks{\texttt{kheeran.naidu@gmail.com}, Qworky Research, London, UK}}

\date{}

\usepackage{algorithm,algorithmicx}
\usepackage{amssymb} 
\usepackage{amsthm}
\usepackage{amsmath} 
\usepackage[noend]{algpseudocode}
\usepackage{algorithm}
\usepackage{thm-restate}
\usepackage{url}
\usepackage{hyperref}
\usepackage{cleveref}

\usepackage[skins,breakable]{tcolorbox}

\newtheorem{theorem}{Theorem}
\newtheorem{lemma}{Lemma}
\newtheorem{definition}{Definition}
\newtheorem{claim}{Claim}
\newtheorem{observation}{Observation}[section]

\newcommand{\cD}{\mathcal{D}}
\newcommand{\cS}{\mathcal{S}}
\newcommand{\cA}{\mathcal{A}}
\newcommand{\bZ}{\mathbb{Z}}
\newcommand{\bE}{\mathbb{E}}
\newcommand{\bP}{\mathbb{P}}
\newcommand{\cT}{\mathcal{T}}

\DeclareMathOperator{\Exp}{\mathbb{E}}

\usepackage[normalem]{ulem}

\begin{document}

\maketitle

\begin{abstract}
We consider the \textsf{Unit Interval Selection} problem in the one-pass random order streaming model. In this setting, an algorithm is presented with a sequence of $n$ unit-length intervals on the line that arrive in uniform random order, one at a time, and the objective is to output (an approximation of) a largest set of disjoint intervals using space linear in the size of an optimal solution. Previous work only considered adversarially ordered streams and established that, within these space constraints, a $(2/3)$-approximation can be achieved in such streams, and this is best possible, in that going beyond such an approximation factor requires space $\Omega(n)$ [Emek et al., TALG'16] .

\vspace{0.1cm}

In this work, we show that an improved expected approximation factor can be achieved if the input stream is in uniform random order, where the expectation is taken over the stream order. More specifically, we give a one-pass streaming algorithm with expected approximation factor $0.7401$ that uses space $O(|OPT|)$, where $OPT$ denotes an optimal solution. We also show that random order algorithms with expected approximation factor above $8/9$ require space $\Omega(n)$, and algorithms that compute a better than $2/3$-approximation with probability above $2/3$ also require $\Omega(n)$ space.

\vspace{0.1cm}

On a technical level, we design an algorithm for the restricted domain $[0, \Delta)$, for some constant $\Delta$, and use standard techniques to obtain an algorithm for unrestricted domains. For the restricted domain $[0, \Delta)$, we run $O(\Delta)$ recursive instances of our algorithm, with each instance targeting the situation where a specific interval of an optimal solution arrives first. We establish the interesting property of our algorithm that it performs worst when the input stream consists solely of a set of independent intervals. It then remains to analyse the algorithm on these simple instances.

\vspace{0.1cm}

Our lower bound is proved via communication complexity arguments, similar in spirit to the robust communication lower bounds established by [Chakrabarti et al., Theory Comput. 2016].
\end{abstract}

\newpage

\section{Introduction}
In the one-pass streaming model of computation, an algorithm is presented with a stream of input items that, a priori, arrive in arbitrary order. The algorithm is tasked with computing a solution while maintaining a memory of size sublinear in the length of the input stream. 

We consider the \textsf{Unit Interval Selection} problem in the streaming model. In this problem, the input stream consists of $n$ closed unit-length intervals $\mathcal{S}$ on the line, presented to the algorithm in arbitrary order, and the objective is to compute a pairwise non-overlapping subset $\mathcal{I} \subseteq \mathcal{S}$ of largest cardinality.

Emek et al. \cite{ehr16} initiated the study of \textsf{Unit Interval Selection} in the one-pass streaming model.
They gave a $2/3$-approximation algorithm that uses $O(\vert OPT \vert)$ words\footnote{A {\em word} is the amount of space required to store a single interval.} of space, where $OPT$ denotes an optimal solution, and showed that $\Omega(n)$ bits of space are necessary for algorithms with even slightly better approximation guarantees, which renders their algorithm space optimal.
Subsequently, Cabello and P{\'{e}}rez{-}Lantero \cite{cp17} gave a simpler algorithm with the same guarantee. We emphasize that these results hold for streams that are in adversarial/arbitrary order.

Given that the adversarial order setting is fully understood, in this paper, we ask whether improved algorithms are possible if the input stream is in uniform random order. Indeed, random order streams have received significant attention, and, for many problems, improvements beyond known lower bounds that hold for adversarial orderings are known, such as \textsf{Quantile Estimation} \cite{gm09}, \textsf{Maximum Matching} \cite{k13,a23,ab21}, \textsf{Bounded Degree Property Testing} \cite{mmps17}, and edge-arrival \textsf{Set Cover} \cite{kka23}. Lower bounds for random order streams were pioneered by Chakrabarti et al. \cite{cjp08}, with subsequent work \cite{ccm08} further developing these ideas. The ideas from \cite{ccm08} have since been employed in various random order streaming lower bound proofs, including \cite{ab21} and \cite{as23}.  The random order setting, which assumes worst-case instances and random order arrivals, is a less pessimistic setting than the arbitrary order setting, which assumes both worst-case instances and worst-case orderings, and is therefore more relevant from a practitioner's point of view.

\subsection{Our Results} 
As our main result, we show that we can go beyond the $2/3$-approximation barrier for adversarial order streams when considering a uniform random arrival order.

\begin{restatable}{theorem}{ALG}\label{thm: alg}
  There is a deterministic one-pass $O(|OPT|)$ words of space streaming algorithm for \textsf{Unit Interval Selection} on random order streams with expected approximation factor $0.7401$, where the  expectation is taken over the random order of the stream.
\end{restatable}

We then complement our algorithm with the following  limitation result:

\begin{restatable}{theorem}{INDEXLB}\label{thm: exp INDEX LB}
  Let $\epsilon>0$ be any small constant. Then: 
  \begin{enumerate}
      \item Any randomized one-pass streaming algorithm for random-order \textsf{Unit Interval Selection} with expected approximation factor $8/9+\epsilon$ must use $\Omega(n)$ bits of space.
      \item For any $\delta > 0$, any randomized one-pass streaming algorithm for random-order \textsf{Unit Interval Selection} with approximation factor $2/3+\delta$ that succeeds with probability at least $2/3+\epsilon$ on any input instance must use $\Omega(n)$ bits of space.
  \end{enumerate}  
  The expectation in (1) and the probability in (2) are taken over both the stream order and the random choices of the algorithm.
\end{restatable}

We note that, in the previous theorem, (1) immediately follows from (2), however, the result stated in (1) is more suitable for comparison to our algorithmic result. Furthermore, the result in (2) explains why our algorithm can only go beyond the $2/3$-approximation barrier that holds for adversarial order streams {\em in expectation}, and not with high constant probability.

Our algorithmic result combined with our lower bound result show that the optimal achievable approximation ratio lies in the interval $[0.7401, 0.\overline{8}]$, and we leave it as a great open question to further narrow this gap.

\subsection{Techniques}
We will now discuss the main ideas behind both our algorithm and our lower bound.

\paragraph{The Algorithm.} We make use of an established technique that, given a one-pass streaming algorithm for unit-length intervals that are all contained in the restricted domain $[0, \Delta)$, for some constant integer $\Delta$, and has approximation factor $\alpha$, produces an algorithm for unrestricted domains with only a constant factor blow-up in space and a factor $(\Delta-1) / \Delta$ blow-up in the approximation factor \cite{ddk23}. We thus see that it is enough to design an algorithm for the restricted domain $[0, \Delta)$, for some large but constant $\Delta$, large enough so as to minimize the impact of the factor $(\Delta-1) / \Delta$ on the final approximation factor. 

To this end, as a key building block of our algorithm, we employ the following strategy: We maintain the left-most interval $I_L$ observed thus far in the stream, and feed the substream of intervals that lie to the right of $I_L$ into a recursive instance of our algorithm. This building block then returns the interval $I_L$ together with the solution produced by the recursive call. To see the strength of this method, denote by $OPT$ an arbitrary but fixed optimal solution. Furthermore, assume that the left-most interval $opt_1 \in OPT$ arrives before any of the other intervals in $OPT$. In this case, the algorithm would have either picked $opt_1$ as the interval $I_L$ or an interval that is even further left as $opt_1$, and the substream fed into the recursive call of the algorithm contains $OPT \setminus opt_1$, i.e., an independent set that is only by one smaller. In other words, the algorithm has acted optimally thus far since one interval $I_L$ was selected and a subinstance that contains an independent set of size $|OPT|-1$ was created whose solution can be combined with the selected interval $I_L$.

Unfortunately, the event that $opt_1$ appears first among the intervals of $OPT$ has a probability of only $1/|OPT|$ (over the random order of the stream). Our aim, however, is to make this approach work, even if any other interval of $OPT$ is the first to arrive among the intervals in $OPT$. First, by symmetry, we see that the approach above also works when the right-most interval of $OPT$ arrives first. Hence, suppose now that the $i$-th (counted from left to right) optimal interval $opt_i \in OPT$ arrives before any of the other optimal intervals, for some $1 < i < |OPT|$. Furthermore, let $\ell$ be the largest integer strictly smaller than the left boundary of $opt_i$, and let $r$ be the smallest integer strictly larger than the right boundary of $opt_i$. Then, we construct two solutions and pick the larger one, as follows. For solution 1, we combine the outputs of two independent runs of our algorithm, one on the domain $[0,\ell)$, and one on the domain $[\ell, \Delta)$. While these runs ignore intervals that intersect with the point $\ell$ and may thus potentially miss a crucial optimal interval that intersects this point, we can now exploit the fact that the interval $opt_i$ is the left-most optimal interval in the domain $[\ell, \Delta)$. Hence, as discussed above, the run on $[\ell, \Delta)$ does not make any mistake when selecting the interval $opt_i$ as the leftmost interval (or an interval within $[\ell, \Delta)$ that lies even further to the left as $opt_i$) into the final solution. For solution 2, we combine the outputs of the two runs of the algorithm on the domains $[0, r)$, and $[r, \Delta)$, respectively, and now exploit the fact that $opt_i$ is the right-most optimal interval in the domain $[0, r)$.

Last, since our algorithm cannot know which interval of $OPT$ will indeed arrive first, we run the above strategy for every integer $1 \le i \le |OPT|$. We observe that this still yields an algorithm with only constant space complexity since the domain $[0, \Delta)$ is of constant size, which further implies that the optimal solution within this domain is also only of constant size $|OPT| \le \Delta - 1$. 

In our analysis, we establish the monotonicity property that adding intervals to the input stream never decreases the size of the output produced by the algorithm.
In other words, the worst-case performance of our algorithm is achieved when the algorithm is run on an independent set -- a special case that could of course be solved optimally by just reporting these intervals. We can thus focus on analysing the performance of the algorithm on this special case. Due to the many recursive calls of the algorithm, we obtain a recursive formula for the solution size obtained. We observe that the approximation factor $\alpha$ of our algorithm decreases as $\Delta$ increases. Recall, however, that we still need to apply the technique mentioned above to go from the restricted domain $[0, \Delta)$ to an unrestricted domain, which incurs a further loss in the approximation factor by a factor of $(\Delta-1) / \Delta$. Using a computer programme, we see that, for $\Delta = 5000$, the final algorithm has an approximation factor of at least $0.7401$, which is provably very close to the best possible choice.
This specific choice of $\Delta$ is discussed further in Appendix \ref{app: different gamma}

\paragraph{The Lower Bound.} 
Our lower bound is obtained via a reduction to the $\textsf{INDEX}_t$ problem in the one-way two-party communication complexity setting, by combining the clique gadgets already used in \cite{ehr16} and \cite{cp17} for adversarial order lower bounds with the random order lower bound ideas by Cormode et al. \cite{ccm08}. 

In the $\textsf{INDEX}_t$ problem, there are two players, denoted Alice and Bob.
Alice holds a vector $X \in \{0,1\}^t$ consisting of $t$ uniform independent random bits, and Bob is given a uniform random index $A \in [t]$. Alice sends a message to Bob and Bob outputs the bit $X[A]$.
The objective is to obtain a communication protocol that communicates a message of smallest possible size. It is known that $\textsf{INDEX}_t$ requires a message of size $\Omega(t)$. Since streaming algorithms can be used as one-way two-party communication protocols (Alice runs the streaming algorithm on her part of the input and sends the memory state to Bob, who then completes the execution of the algorithm on his part), lower bounds on the message size translate and yield lower bounds on the space used by streaming algorithms. 

The $\textsf{INDEX}_t$ problem can be reduced to \textsf{Unit Interval Selection} as follows.
Using the bits in $X \in \{0, 1\}^t$, Alice constructs a \textit{clique} of $t$ mutually overlapping intervals, i.e., a stack of intervals such that, for every $i > 1$, the $i$-th interval is slightly further to the right than the $(i-1)$-th interval. The bits $X[i]$ are encoded in this construction in that, depending on the value of $X[i]$, the $i$-th interval is ever so slightly moved, which, however, does not affect the logic behind the construction. Crucially, an algorithm that outputs the $j$-th interval of the clique stack, for any $j$, must know the value $X[j]$. Then, based on the index $A$, Bob constructs two {\em wing intervals}, which surround, but are independent of, the interval corresponding to $X[A]$.
The key property is that every clique interval, other than the one corresponding to $X[A]$, must intersect exactly one of these wing intervals, or, in other words, the only independent set of size $3$ consists of the wing intervals together with the $A$-th interval of the stack. This in turn implies that any algorithm on this construction that computes a better than $2/3$-approximation must report the interval corresponding to $X[A]$, and its location allows recovering the bit $X[A]$.

When bringing this construction into the random-order setting, for the resulting construction to be hard it is necessary that the two wing intervals appear after the $A$-th clique interval since otherwise the positions of the wing intervals reveal the index $A$ and an algorithm could then simply collect the $A$-th clique interval when it arrives, obtain an optimal solution to the interval selection problem, and solve the underlying $\textsf{INDEX}_t$ instance. We, however, observe that, under a uniform random arrival order, with probability $1/3$, the two wing intervals indeed arrive after the $A$-th clique interval, and we prove that this case indeed remains hard to solve. Our expected approximation factor lower bound of $8/9$ is then explained by the fact that, with probability $1/3$, an algorithm can only obtain a solution of size $2$, and with probability $2/3$, an algorithm can obtain a solution of optimal size $3$, which results in an expected solution size $1/3 \cdot 2 + 2/3 \cdot 3 = 8/3$, or an approximation factor of $(8/3) / 3 = 8/9$.

On a technical level, our approach is a reduction via the use of shared public randomness in the spirit of the argument from \cite{ccm08}.
Alice and Bob  sample a shared random permutation of the intervals that are used in the hard instance.
Alice constructs each clique interval which arrives before a wing interval via the bits in $X$, and then Bob constructs the wing intervals via $A$, and also the remaining clique intervals via sampling public random variables.
This results in a construction that reflects the bit $X[A]$ of the input to $\textsf{INDEX}_t$ precisely when the interval corresponding to $X[A]$ arrives before either wing interval, and this occurs with probability $1/3$, as desired, which yields our lower bound results. 

\subsection{Further Related Work}
Besides unit-length intervals, intervals of arbitrary lengths have also been considered. Emek et al. \cite{ehr16} showed that, in adversarial order streams, a $1/2$-approximation for arbitrary-length intervals can be achieved with space $O(|OPT|)$, and they also proved that algorithms with improved approximation factor must use space $\Omega(n)$. Cabello and P\'{e}rez-Lantero \cite{cp17} gave a simpler algorithm for this case. More recently, Dark et al. \cite{ddk23} explored algorithms for weighted interval streams and streams that allow for deletions. In a different line of work, the problem of estimating the size of a largest disjoint set of intervals has also been considered \cite{cp17,bcw20}. Besides intervals, other geometric objects, such as squares and rectangles \cite{cdk19}, and segments and 2-intervals \cite{bko22}, have also been explored.

\textsf{Interval Selection} has also been studied in the streaming sliding window model, where the objective is to maintain a solution on the $L$ most recent intervals. Alexandru and Konrad \cite{ak24} showed that \textsf{Interval Selection} is harder in the sliding window model than in the one-pass streaming setting in that, for both unit-length intervals and arbitrary-length intervals, algorithms cannot achieve the same approximation guarantees as in the streaming setting when using space $O(|OPT|)$.

\subsection{Outline}
We provide definitions and state the hardness of the $\textsf{INDEX}_t$ communication problem that we use for our lower bound result in our preliminaries section, Section~\ref{sec: prelims}. Then, in Section~\ref{sec: alg}, we present and analyse our algorithm, and in Section~\ref{sec: LBs}, we give our lower bound result. Finally, we conclude in Section~\ref{sec: conclusion}.

We further discuss the exact approximation factor achieved in Appendix \ref{app: different gamma}.

\section{Preliminaries}\label{sec: prelims}

We use $n$ to denote the length of the input  stream and $[k]$ to denote the set $\{0,\dots,k-1\}$.

\subsection{Interval Selection}

In this work, we consider closed intervals $I = [a, b]$, where $a < b$. We mostly focus on intervals of {\em unit length}, i.e., when $b = a +1$. We say that $a$ is the {\em left endpoint} of $I$ and $b$ is the {\em right endpoint} of $I$. We say that an interval $I=[a_I, b_I]$ is {\em further left} than an interval $J = [a_J, b_J]$ if its left endpoint is strictly smaller, i.e., $a_I < a_J$ holds. Similarly, $I$ is {\em further right} of $J$  if its right endpoint is strictly larger, i.e., $b_I > b_J$ holds.

Two intervals $I, J$ are {\em independent} if $I \cap J = \varnothing$, and a set of intervals $\mathcal{I}$ is an {\em independent set} if the intervals contained in $\mathcal{I}$ are  pairwise independent.
The \textsf{Unit Interval Selection} problem asks for an independent set of intervals  of largest size, where we also assume that all intervals have length $1$.

We use $OPT$ to denote an arbitrary but fixed independent set of largest size, and we write $\alpha = \alpha(\cdot)$ to denote the size of a largest independent set, where the parameter passed on to $\alpha(\cdot)$ is either a stream or a set or intervals. For a parameter $0 < \beta \le 1$, we say that an independent set $\mathcal{I}$ is a {\em $\beta$-approximation} if $\vert\mathcal{I}\vert\geq \beta \cdot \vert OPT \vert$.

We also assume that any interval can be stored in a single word of space.
Under the plausible assumption that the endpoint can be encoded in $O(\log n)$ bits of space, this yields $O(\log n)$ space overhead.

\subsection{Communication Complexity}

We use the standard one-way two-party communication complexity setting, first introduced by \cite{y79}.
Our lower bounds will then be derived from a reduction to the well known $\textsf{INDEX}_t$ problem.
We refer the reader to \cite{ry20} for a comprehensive overview of communication complexity.

In this setting, we have two parties that we denote by Alice and Bob.
They each hold a portion of the input, $X_A$ and $X_B$, respectively, and their goal is to compute/approximate a function $f(X_A, X_B)$ on their inputs. To this end,
Alice sends a single message to Bob, $M(X_A)$, and after receiving the message, Bob outputs $f(X_A, X_B)$ or a suitable approximation thereof.
Then, the size of Alice's message is the relevant notion of complexity, and the goal is to send a message that is as small as possible. Alice and Bob can make use of infinite sequences of both private and public/shared random bits. 

\begin{definition}
For an integer $t$, $\textsf{INDEX}_t$ is the one-way two-party communication game where Alice holds a vector $X\in\{0,1\}^t$ and Bob holds an index $A\in[t]$. Alice sends a single message to Bob who then outputs the value $X[A]$.
\end{definition}
Let $\mu$ denote the uniform distribution over $\textsf{INDEX}_t$ instances. This is the distribution where $X$ is sampled uniformly at random from $\{0,1\}^t$ and $A$ is sampled independently from $X$ and uniformly at random from its support $[t]$.

It is well-known that solving $\textsf{INDEX}_t$ requires a message of size $\Omega(t)$, even in a distributional sense when the input is chosen from $\mu$, and also regardless of any public randomness used. 

\begin{theorem}[$\textsf{INDEX}_t$ Hardness (e.g., \cite{ry20})]\label{thm: INDEX hardness}
  Let $\delta > 0$ be any constant.
  Let $P$ be a one-way communication protocol for $\textsf{INDEX}_t$.
  Let $R$ denote any public randomness used by $P$.
  If 
  \[
    \mathbb{P}_{\mu, R}(\text{$P$ outputs the correct answer}) \geq 1/2 + \delta,
  \]
  then $P$ must send a message of size $\Omega(t)$ bits.
\end{theorem}

\section{Random-Order Unit Interval Selection Algorithm} \label{sec: alg}

We now prove Theorem \ref{thm: alg}. To this end, we first present our algorithm $\mathcal{A}$ (see Algorithm~\ref{alg: alg}), and then analyse it.

\begin{algorithm}
    \caption{Algorithm $\cA$}\label{alg: alg}
    \textbf{Input:} A domain $\cD = [a,b)$, where $a,b \in \bZ$, and $a < b$
    \newline
    \textbf{Initialisation:}
        \begin{algorithmic}[1]
        \State // the recursive instances corresponding to each split point
        \For{\textbf{each} integer $i \in \{a+1,\dots,b-1\}$}
            \State // Initialise the left and right most intervals closest to each integer position
            \State $L_i \gets \emptyset$
            \State $R_i \gets \emptyset$
            \State // Initialise the recursive instances to the left of each integer position
            \State $\cA^L_i \gets \text{new instance $\cA([a, i))$}$
                \State $\cT^L_i \gets \text{new instance $\cA([a, i))$}$
            \State // Initialise the recursive instances to the right of each integer position
            \State $\cA^R_i \gets \text{new instance $\cA([i, b))$}$
                \State $\cT^R_i \gets \text{new instance $\cA([i, b))$}$
        \EndFor
    \end{algorithmic}
    \vspace{2mm}
    \hrule
    \vspace{2mm}
  
    \textbf{During the Stream:}
    \newline
    \textbf{Input:} A stream $\cS$ of unit-length intervals, each of which is fully contained in $[a,b)$.
    \begin{algorithmic}[1]
        \For{\textbf{each} interval $I \in \cS$}
            \For{\textbf{each} integer $i \in \{a+1,\dots,b-1\}$}
                \State // Update $R_i$ and feed $I$ into the corresponding right recursive instances
                \If{$I \subseteq [i, b)$}
                    \State Feed $I$ into $\cT^R_i$
                    \If{$R_i = \emptyset$ \textbf{or} ($R_i \neq \emptyset$ and $I$ is further left than $R_i$)}
                        \State $R_i \gets I$
                    \EndIf
                    \If{$R_i \neq \emptyset$, $I$ is independent of $R_i$, and $I$ is further right than $R_i$}
                        \State Feed $I$ into $\cA^R_i$
                    \EndIf
                \EndIf

            \State // Update $L_i$ and feed $I$ into the corresponding left recursive instances
                \If{$I \subseteq [a, i)$}
                    \State Feed $I$ into $\cT^L_i$
                    \If{$L_i = \emptyset$  or ($L_i \neq \emptyset$ and $I$ is further right than $L_i$)}
                        \State $L_i \gets I$
                    \EndIf
                    \If{$L_i \neq \emptyset$, $I$ is independent of $L_i$, and $I$ is further left than $R_i$}
                        \State Feed $I$ into $\cA^L_i$
                    \EndIf
                \EndIf
            \EndFor
        \EndFor
    \end{algorithmic}

    \vspace{2mm}
    \hrule
    \vspace{2mm}
  
    \textbf{Output:}
    \begin{algorithmic}[1]
        \State The largest set among
        \[
            OUT\left(\cT^L_i\right) \cup R_i \cup OUT\left(\cA^R_i\right)
        \quad \text{ and } \quad
            OUT\left(\cA^L_i\right) \cup L_i \cup OUT\left(\cT^R_i\right)
        \]
        for each integer $i \in \{a+1,\dots,b-1\}$
    \end{algorithmic}
\end{algorithm}

Our algorithm proceeds as follows. It is parametrised by two integers $a,b$ with $a < b$, and operates under the assumption that all intervals fed into the algorithm are fully contained in the domain $\mathcal{D} = [a, b)$. Our overall goal is to design an algorithm that operates on the domain $[0, \Delta)$, for some integer $\Delta$. To accommodate for this, we will see later that, for technical reasons, we need to invoke our algorithm with parameters $a=-1$ and $b = \Delta+1$.

Next, our algorithm considers all integers in $\{a+1, \dots, b-1\}$ as {\em split points}. For each such split point $i \in \{a+1, \dots, b-1\}$, the algorithm maintains in variables $L_i$ and $R_i$ the closest interval to $i$ that is located left of $i$ and the closest interval to $i$ that is located right of $i$, respectively. It also executes four recursive calls of the algorithm associated with $i$. The recursive calls $\mathcal{T}_i^L$ and $\mathcal{T}_i^R$ are instantiated on the domains $[a, i)$ and $[i, b)$, respectively, and all intervals that fall into these domains are fed into these algorithms. The recursive calls  $\mathcal{A}_i^L$ and $\mathcal{A}_i^R$ are also instantiated on the domains 
$[a, i)$ and $[i, b)$, however, only intervals that are left of $L_i$ (and within the domain $[a, i)$) are fed into $\mathcal{A}_i^L$, and intervals that right of $R_i$ (and within the domain $[i, b)$) are fed into $\mathcal{A}_i^R$.

For each split point $i$, two potential outputs are considered: The output generated by $\mathcal{T}_i^L$, combined with the output of $\mathcal{A}_i^R$ as well as the interval $R_i$, and the output generated by $\mathcal{A}_i^L$, combined with the output of $\mathcal{T}_i^R$ as well as the interval $L_i$. We note that, by construction, both these potential outputs form independent sets. The output of the algorithm then is the largest output of this kind when considering all split points.

In the following, we write $OUT(\cA(\cS))$ to denote the output produced by algorithm $\mathcal{A}$ when run on the input stream  $\cS$.

\subsection{Monotonicity Property}

We first prove that removing any interval from the input stream cannot increase the size of the output  regardless of the order of the stream.

\begin{lemma}\label{lem: technical}
    Let $\cA$ be an instance of Algorithm \ref{alg: alg} instantiated with integers $a < b$.
    Let $\cS$ be any input stream (not necessarily random-order) whose intervals are fully contained in $[a,b)$, and let $\cS'$ be any substream of $\cS$.
    Then $\vert OUT(\cA(\cS')) \vert \leq \vert OUT(\cA(\cS)) \vert$.
\end{lemma}
\begin{proof}
    Our proof is via induction over the quantity $b - a$.
    We first show that the recursion tree will always have finite size.

    \begin{claim}
        For any value of $b - a > 0$, the recursion tree of $\cA$ has finite size.
    \end{claim}
    \begin{proof}
        Each recursive subproblem is instantiated before any intervals arrive. Thus, the size of the recursion tree depends only on $b - a$ and not on the input stream.
        Each recursive instance has a domain that is by one smaller than the domain of its parent, and an instance with domain size at most 1 creates no subproblems.
    \end{proof}

    We now proceed with the induction.

    \paragraph*{Base Case:}
        Suppose $\cA$ is an instance of Algorithm \ref{alg: alg} such that $b - a = 1$.
        Then $\vert OUT(\cA(\cS')) \vert \leq \vert OUT(\cA(\cS)) \vert$.
    \begin{proof}
        If $b - a = 1$, then no unit length intervals fit within $[a, b)$.
        Therefore, $\cS$ and $\cS'$ must both be empty streams, and so $\vert OUT(\cA(\cS')) \vert = \vert OUT(\cA(\cS)) \vert = 0$.
    \end{proof}

    \paragraph*{Inductive Hypothesis:}
        Suppose the lemma holds for all instances of $\cA$ such that $b - a \leq k-1$.

        We now consider when $b - a = k$.

    \paragraph*{Inductive Step:}
        By construction, the output of $\cA$ is formed by combining the outputs of its recursive children, together with either $L_i$ or $R_i$, for some $i$.
        We will first argue that if any $L_i$ or $R_i$ stores an interval under $\cS'$, then it must also store an interval under $\cS$.
        We will then argue that all of the recursive children must output solutions at least as large under $\cS$ compared to $\cS'$.
        This is sufficient to prove the claim.
 
        First, W.L.O.G., suppose some $L_i$ is non-empty under $\cS'$.
        Then, there must exist some interval in $\cS'$ which is contained in the domain $[a,i)$.
        Then, this same interval must also be contained in $\cS$ ($\cS$ is a superstream of $\cS'$), and so $L_i$ must also be non-empty under $\cS$.
        This argument holds for all relevant integers $i$, and also for each $R_i$.

        We now turn to the substreams fed into the recursive children of $\cA$, and begin with the instances $\cT_i^L$ and $\cT_i^R$.
        W.L.O.G. consider $\cT_i^L$ for some relevant integer $i$.
        Then, by construction, the size of the domain of $\cT_i^L$ is at most $k - 1$.
        Next, let $\cS'_i$ denote the substream of intervals fed into $\cT_i^L$ when $\cS'$ is fed into $\cA$, and $\cS_i$ the substream fed in when $\cS$ is fed into $\cA$.
        Then, $\cS_i'$ is precisely the substream of intervals in $\cS'$ that are fully contained in $[a,i)$.
        Then, every interval in $\cS_i'$ is also contained in $\cS'$ and so $\cS$, and therefore every interval in $\cS_i'$ is also contained in $\cS_i$.
        This gives us that $\cS_i'$ is a substream of $\cS_i$.
        We can now apply the inductive hypothesis to obtain
        \[
            \left\vert OUT\left(\cT_i^L(\cS_i')\right) \right\vert \leq \left\vert OUT\left(\cT_i^L(\cS_i)\right) \right\vert.
        \]
        The same argument applies for every relevant integer $i$, and also for each $\cT_i^R$.

        Finally, we now consider the instances $\cA_i^L$ and $\cA_i^R$.
        W.L.O.G. consider $\cA_i^L$ for some relevant integer $i$.
        Then, by construction, the size of the domain of $\cA_i^L$ is at most $k - 1$.
        Next, let $\cS_i'$ denote the substream of intervals fed into $\cA_i^L$ when $\cS'$ is fed into $\cA$, and $\cS_i$ the substream fed in when $\cS$ is fed into $\cA$.
        Then, these are each the substreams of $\cS'$ and $\cS$ of the intervals fully contained in $[a,i)$, which are also further left and independent of $L_i$ at the time each interval arrives.
        We now observe that whenever an interval in $\cS_i'$ arrives in $\cS$, the interval stored in $L_i$ must be at least as far right as when the same interval arrives in $\cS'$ - by the construction of the algorithm.
        Therefore, every interval in $\cS_i'$ is also contained in $\cS_i$, giving us that $\cS_i'$ is a substream of $\cS_i$.
        We can now apply the inductive hypothesis to obtain
        \[
            \left\vert OUT\left(\cA_i^L(\cS_i')\right) \right\vert \leq \left\vert OUT\left(\cA_i^L(\cS_i)\right) \right\vert.
        \]
        The same argument applies for every relevant integer $i$, and also for each $\cA_i^R$.

    This concludes the proof of the Lemma.
\end{proof}

\subsection{Bounding the Expected Approximation Factor}\label{subs: expected approx}

We now derive a bound on the expected approximation factor achieved by Algorithm \ref{alg: alg}. To this end, let $\Delta$ be an arbitrary but fixed integer, and we consider the performance of our algorithm when all intervals lie within the domain $[0, \Delta)$. For technical reasons, we run our algorithm $\mathcal{A}$ on the domain $[-1, \Delta+1)$ and analyse the performance of this instantiation of the algorithm.

For a set of intervals $\mathcal{I}$ that constitutes the input instance, we denote by $\Pi(\mathcal{I})$ the set of permutations of $\mathcal{I}$. Observe that the input stream then constitutes a uniform random element of $\Pi(\mathcal{I})$. 

We are interested in the expected solution size produced by the algorithm as a function of $\alpha(\mathcal{I})$, i.e., the optimal solution size.
For this purpose, for any integer $x \ge 0$, we define the following quantity:
\begin{align}
    out(x) = \min_{\text{instance $\mathcal{I}$ with $\alpha(\mathcal{I}) = x$}} \Exp_{S \sim \Pi(\mathcal{I})} |OUT(\mathcal{A}(S))| \ ,
\end{align}
i.e., the worst-case expected performance of the algorithm on instances with optimal solution size $x$.
For $x < 0$, define $opt(x) = 0$.

We first observe that it suffices to analyse the performance of our algorithm on streams consisting solely of independent intervals. 

\begin{observation}\label{obs: worst case stream}
    Let $\mathcal{I}$ be any input instance, and let $OPT \subseteq \mathcal{I}$ denote the optimal intervals in $\mathcal{I}$. Then, 
    $$\Exp_{S_{OPT} \sim \Pi(OPT)} |OUT(\mathcal{A}(S_{OPT}))| \le \Exp_{S_{\mathcal{I}} \sim \Pi(\mathcal{I})} |OUT(\mathcal{A}(S_{\mathcal{I}}))| \ . $$    
\end{observation}
\begin{proof}
    For a permutation $S_{\mathcal{I}}$ of the intervals $\mathcal{I}$, we denote by $OPT(S_{\mathcal{I}})$ the substream of optimal intervals. Then, since $OPT(S_{\mathcal{I}})$ is a substream of $S_{\mathcal{I}}$, we can use Lemma~\ref{lem: technical} to  obtain:
    \[
        \Exp_{S_{\mathcal{I}} \sim \Pi(\mathcal{I})} |OUT(\mathcal{A}(OPT(S_{\mathcal{I}})))| \leq \Exp_{S_{\mathcal{I}} \sim \Pi(\mathcal{I})} |OUT(\mathcal{A}(S_{\mathcal{I}}))|.
    \]
The result then follows from the observation that $S_{OPT}$ and $OPT(S_{\mathcal{I}})$ are identically distributed.    
\end{proof}

We will thus assume from now on that the input instance $\mathcal{I}$ is a collection of disjoint intervals.
We then take the input stream to be $\cS \sim \Pi(I)$ for the remainder of this subsection, and omit writing that $\cS$ is the input to $\cA$.

In the following, we will establish lower bounds on $out(x)$, for every $x \ge 0$. To this purpose, we treat the cases $x \in \{0, 1, 2\}$ explicitly, and then provide a recurrence relation for larger values of $x$. 

\begin{restatable}{lemma}{lema}\label{lem: OPT leq 2}
  For each $x\in\{0,1,2\}$, $out(x) = x$.
\end{restatable}
\begin{proof}
    First, $\alpha = 0$ implies the input stream has no intervals, which is trivial.
    Next, if $\alpha = 1$, then the input stream consists of a single interval. This interval will be picked up by $R_1$ since $R_1$ constitutes the left-most interval in the domain $[0, \Delta)$. 

    Now suppose that $\alpha = 2$, which implies that the input consists of two intervals $opt_1$ and $opt_2$ with $opt_1$ further left of $opt_2$.
    If $opt_1$ arrives before $opt_2$, then $R_1$ will store $opt_1$, and $opt_2$ will be fed into $\cA^R_1$ and returned by this recursive call.
    The interval $opt_2$ will then be stored as the right most interval fed into $\cA^R_1$, ensuring that $\vert OUT(\cA^R_1) \vert \geq 1$.
    We thus obtain that $\vert R_0 \cup OUT(\cA^R_0) \vert = 2$.
    Then, by the construction of the algorithm, we get
    \[
        \vert OUT(\cA) \vert \geq \vert \{ R_1 \} \cup OUT(\cA^R_1) \vert \geq 1 + 1 = 2.
    \]
    Similarly, if $I_2$ arrives first then $L_{\Delta+1}$ and the output of $\cA^L_{\Delta+1}$ constitutes a candidate solution of size $2$.
\end{proof}

Next, we consider the case $\alpha \geq 3$.
We write
\[
  OPT = \{opt_1,\dots,opt_\alpha\},
\]
to denote the optimal intervals ordered from left to right,
and consider the cases for which interval from $OPT$ arrives first.
We then fix some $i\in\{1,\dots,\alpha\}$, and suppose that $opt_i$ arrives in the input stream before any other interval in $OPT$.
Let $opt_i = [f,f+1]$, for some real $f$, and let 
\[
  r = \lfloor f+2 \rfloor, \text{ and } l = \lfloor f \rfloor
\]
denote the smallest integer larger than $f+1$ that does not intersect with $opt_i$, and the largest integer smaller than or equal to $f$ which may only intersect with $opt_i$'s left endpoint, respectively.
Note that when $f$ is an integer, we have $l = f$, but $r = f + 2$. This choice is due to the fact that every recursive instantiation of our algorithm is run on a domain with a closed left endpoint and an open right endpoint.

\begin{restatable}{lemma}{lemb}\label{lem: with opt_i}
    Conditioned on $opt_i = [f, f+1]$ arriving first, both of the following hold:
    \begin{enumerate}
        \item $\bE \left[ \left\vert R_l \cup OUT\left(\cA^R_l\right) \right\vert \ | \ opt_i \text{ arrives first} \right] = 1 + out(\alpha - i)$, and
        \item $\bE \left[ \left\vert L_r \cup OUT\left(\cA^L_r\right) \right\vert \ | \ opt_i \text{ arrives first} \right] = 1 + out(i - 1)$.
    \end{enumerate}
\end{restatable}
\begin{proof}
    We only prove 1) -- the proof of 2) is similar and omitted.
    
    By definition, the interval $opt_i$ is further right and independent of $l$, and so will be stored in $R_l$.
    Next, observe that each of the intervals $opt_{i+1},\dots,opt_{\alpha}$ are fed into $\cA^R_l$ in uniform random-order. 
    Then, by definition of $out(.)$, we obtain $\bE\vert OUT(\cA^R_l) \vert \ge out(\alpha - i)$.
    Finally, combining $R_l$ and $\cA_R^l$ via linearity of expectation yields the desired bound.

    \textit{Remark:} We note that each interval being contained in the domain $[0,\Delta)$ implies that $0 \leq l$ and also $r \leq \Delta$.
    Recall that we run algorithm $\mathcal{A}$ on the domain $[-1, \Delta+1)$, which implies that both the instances $\cA^R_l$ and $\cA^L_r$ are guaranteed to exist.
\end{proof}

The next lemma follows by a similar proof.

\begin{restatable}{lemma}{lemc}\label{lem: without opt_I}
    Conditioned on $opt_i = [f,f+1]$ arriving first, both of the following hold:
    \begin{enumerate}
        \item $\bE \left[ \left\vert OUT\left(\cT^R_r\right) \right\vert \ | \ opt_i \text{ arrives first} \right] \geq out(\alpha - i - 1)$, and
        \item $\bE \left[ \left\vert OUT\left(\cT^L_l\right) \right\vert \ | \ opt_i \text{ arrives first} \right] \geq out(i - 2)$.
    \end{enumerate}
\end{restatable}
\begin{proof}
    We only prove 2) -- the proof of 1) is similar and omitted. 
    
    By construction, the substream $\mathcal{S}'$ fed into $\cT^L_l$ consists of the intervals in $\{opt_1,\dots,opt_{i-1}\}$ that do not overlap with the point $l$.
    Then, since each interval is of unit length, and also $f - l \leq 1$ holds, $\mathcal{S}'$  contains either the $i-1$ or $i-2$ left-most intervals of $OPT$, and these are fed into $\cT^L_l$ in a uniform random order.
    We thus obtain $\bE\vert OUT(\cT^L_l) \vert \geq \min\{out(i-1), out(i-2)\}$.
    Finally, by Lemma \ref{lem: technical}, $out(\cdot)$ is a non-decreasing function, which implies that $out(i-1) \ge out(i-2)$ holds, and completes the proof.
\end{proof}

We now combine these lemmas to form a recurrence for $out(\cdot)$.
In the following, for convenience, we write the maximum function over multiple lines.

\begin{restatable}{lemma}{lemd}\label{lem: combined output bound}
    Conditioned on $opt_i$ arriving first, the following holds:
    \[
        \bE\left[\vert OUT(\cA)\vert \mid \text{$opt_i$ arrives first} \vert\right] \geq 1 + \max\left\{ \begin{array}{lr} out(i - 1) + out(\alpha - i -1)\\ out(\alpha - i) + out(i - 2)\end{array} \right\}. 
    \]    
\end{restatable}
\begin{proof}
     The proof follows almost immediately from Lemmas~\ref{lem: with opt_i} and \ref{lem: without opt_I}:    
    \begin{align*}
        \bE\vert OUT(\cA) \mid \text{$opt_i$ arrives first} \vert 
                    &\geq \bE\left[\max\left\{ \begin{array}{lr} \vert OUT\left(\cA^L_r\right) \cup L_r \cup OUT\left(\cT^R_r\right) \vert\\ \vert OUT\left(\cT^L_l\right) \cup R_l \cup OUT\left(\cA^R_l\right) \vert \end{array} \right\}\right]\tag{Construction of the algorithm.}\\
                    &\geq \max\left\{ \begin{array}{lr} \bE\vert OUT\left(\cA^L_r\right) \cup L_r \cup OUT\left(\cT^R_r\right) \vert\\ \bE\vert OUT\left(\cT^L_l\right) \cup R_l \cup OUT\left(\cA^R_l\right) \vert \end{array} \right\}\tag{Jensen's inequality.}\\
                    &\geq \max\left\{ \begin{array}{lr} 1 + out(i - 1) + out(\alpha - i -1)\\ 1 + out(\alpha - i) + out(i - 2)\end{array} \right\}.\tag{Lemmas \ref{lem: with opt_i} and \ref{lem: without opt_I}.}\\
                    &\geq 1 + \max\left\{ \begin{array}{lr} out(i - 1) + out(\alpha - i -1)\\ out(\alpha - i) + out(i - 2)\end{array} \right\}
    \end{align*}
\end{proof}

In the following lemma, we obtain a bound on the expected approximation factor of our algorithm.

\begin{lemma}\label{lem: alg restricted domains}
    The expected approximation factor achieved by $\cA$ on intervals that lie in the domain $[0, \Delta)$ is bounded from below by
    \begin{align*}
        \bE&[\text{approximation factor}]\\
           &\geq \min_{\alpha\in\{1,\dots,\Delta-1\}}\left\{ \frac{1}{\alpha} + \frac{1}{\alpha^2}\cdot\sum_{i=1}^{\alpha} \max\left\{ \begin{array}{lr} out(i - 1) + out(\alpha - i -1)\\ out(\alpha - i) + out(i - 2)\end{array} \right\} \right\},
    \end{align*}
    where $out(2) = 2$, $out(1) = 1$, and $out(y) = 0$ for all $y \leq 0$.
\end{lemma}
\begin{proof}
    Let $\alpha \in \{1, \dots, \Delta-1\}$ and consider an input instance $\mathcal{I}$ that establishes the minimum in the definition of $out(\alpha)$. We consider a run of $\mathcal{A}$ on $\mathcal{I}$. Then:
    \begin{align*}
        out(\alpha) & = \bE \left[ \vert OUT(\mathcal{A}) \vert \right] \\
                  &= \sum_{i=1}^{\alpha} \bP(opt_i \text{ arrives first}) \cdot \bE\left[\vert OUT(\mathcal{A}) \vert \mid opt_i \text{ arrives first} \right]\tag{Law of total expectation.}\\
                  &= \sum_{i=1}^{\alpha} \frac{1}{\alpha} \cdot \bE\left[\vert OUT(\mathcal{A}) \vert \mid opt_i \text{ arrives first} \right]\tag{Uniform random arrival order.}\\
                  &\geq \frac{1}{\alpha} \cdot \sum_{i=1}^{\alpha}  \left[ 1 + \max\left\{ \begin{array}{lr} out(i - 1) + out(\alpha - i -1)\\ out(\alpha - i) + out(i - 2)\end{array} \right\} \right]\tag{By Lemma \ref{lem: combined output bound}.}\\
                  &= 1 + \frac{1}{\alpha}\cdot\sum_{i=1}^{\alpha} \max\left\{ \begin{array}{lr} out(i - 1) + out(\alpha - i -1)\\ out(\alpha - i) + out(i - 2)\end{array} \right\}
    \end{align*}
    In particular, this gives us a recurrence relation for $out(\cdot)$, which we will later use to compute values numerically.
    The expected approximation factor is then simply this quantity multiplied by $1/\alpha$.

    Lastly, when our algorithm is run on intervals in the domain $[0, \Delta)$, then we only know that $\alpha \in \{1, \dots, \Delta-1 \}$. Hence, to bound the approximation factor of our algorithm on such instances, we consider all possible values of $\alpha \in \{1, \dots, \Delta-1 \}$ and pick the minimum approximation factor for these cases, which establishes the lemma. 
\end{proof}

\subsection{Bounding the Space Used}

In this section, we establish a bound on the space used by our algorithm. 

\begin{restatable}{lemma}{leme}\label{lem: space used}
    Let $\cA$ be an instance of Algorithm \ref{alg: alg} run on domain boundaries $a < b$ with $b-a = k$. 
    Then, $\cA$ uses $O(4^k k^k)$ words of space.
\end{restatable}
\begin{proof}
    By construction, $\cA$ consists of the following:
    \begin{itemize}
        \item $4k - 4$ recursive instances of the algorithm, each on input domains of length at most $k - 1$,
        \item at most $2k - 2$ many stored intervals (in each $L_i$ and $R_i$).
    \end{itemize}
    Then, letting $S(k)$ denote the space used by an instance of Algorithm \ref{alg: alg} on a domain of length $k$, we obtain the recurrence
    \begin{align*}
        S(k) &\leq (4k - 4)\cdot S(k-1) + 2k - 2 \ , 
    \end{align*}
    which implies $S(k) = O(4^k k^k)$.
\end{proof}

\subsection{Going to Unrestricted Domains}

We now use the shifting window idea, originally by Hochbaum and Mass \cite{hm85}, and later applied by \cite{cp17} and \cite{ddk23}, that allows going from restricted to unrestricted domains. 

Our lemma below is a straightforward reformulation to random-order streams of a corresponding lemma given in \cite{ddk23} for adversarial order streams.
The key difference is that we now also use linearity of expectation to deal with the random order of the stream, as opposed to the previously considered adversarial order streams.
The proofs are otherwise quite similar.

\begin{restatable}{lemma}{URDOMAINS}\label{lem: unrestricted domain}
  Let $\Delta \ge 2$ be an integer.
  Let $\mathcal{A}$ be any (deterministic) algorithm for random-order unit-length interval selection on the restricted domain $[0,\Delta)$ that computes an $\alpha$-approximation in expectation (over the random order of the stream) using space $O(s)$.
  Then there is a (deterministic) algorithm for random-order unit-length interval selection on unrestricted domains that computes an $\alpha\cdot(\Delta-1)/\Delta$-approximation  in expectation (over the random order of the stream) using space $O(s\cdot\Delta\cdot\vert OPT \vert)$ that uses $\mathcal{A}$ as a black box.
\end{restatable}
\begin{proof}
  Let $\mathcal{A}$ be any deterministic algorithm for random-order unit-length interval selection on the restricted domain $[0,\Delta)$ that computes an $\alpha$-approximation in expectation (over the random order of the stream) using space $O(s)$.
  We use $\mathcal{A}$ as a black-box to present an algorithm for unrestricted domains.

  To this end, for all $i\in\mathbb{Z}$, we first define the window $W_i^{\Delta} = [i, i+\Delta)$. Initially, we mark every window $\{W_i^{\Delta} : i\in\mathbb{Z}\}$ as \textsf{inactive}.
  For each interval $I$ arriving in the stream, let $\mathcal{W}_I = \{W_i^{\Delta} : I\subseteq W_i^{\Delta}\}$ be the set of windows that fully contain $I$.
  For each $W\in\mathcal{W}_I$, if $W$ is marked \textsf{inactive}, create a new instance of $\mathcal{A}$ corresponding to $W$, and feed $I$ into $\mathcal{A}$.
  Then mark $W$ as \textsf{active}.
  Otherwise, if $W$ is already marked \textsf{active} then there is already an instance of $\mathcal{A}$ running that corresponds to $W$.
  In this case, we feed the interval $I$ into the instance of $\mathcal{A}$.
  We note that, as a technicality, whenever an interval $I=[a,b]$ is to be fed into an instance $\mathcal{A}$ corresponding to a window $W_i^{\Delta}$ then
  $\mathcal{A}$ is given the transformed interval $I' = [a-i,b-i]$ so that $\mathcal{A}$'s input is contained within the restricted domain $[0,\Delta)$.

  For each window $W$, let $OUT(W)$ denote the output from the instance of $\mathcal{A}$ corresponding to $W$ at the end of the stream.
  We implicitly transform each interval in $OUT(W)$ back so that they correspond to intervals on the unrestricted domain.
  At the end of the stream output a maximum independent set of intervals from
  \[
    \bigcup_{\substack{i\in\mathbb{Z}\\ \text{$W_i^{\Delta}$ is $\textsf{active}$}}} OUT(W_i^{\Delta})
  \]
  We optimize the space used by the algorithm by only storing the set of windows that are marked as \textsf{active}, i.e. each window is implicitly marked as \textsf{inactive} until the time it is explicitly marked as \textsf{active}.
  Then, a window marked as \textsf{inactive} uses zero space.

  The space used by the algorithm is precisely the space used by windows marked as \textsf{active}, and each window uses space $O(s + 1) = O(s)$ (to store $\mathcal{A}$ as well as mark it \textsf{active}).
  For each window $W_i^{\Delta}$, define $W_i^{\Delta}(L) = [i-1,i-1+\Delta)$ and $W_i^{\Delta}(R) = [i+1,i+1\Delta)$ to be the windows shifted by one to the left and the right, respectively.
  Then define
  \begin{align*}
    W(OPT) &= \{W_i^{\Delta} : \exists I\in OPT\text{ s.t. }I\subseteq W_i^{\Delta}\},\\
    W_L(OPT) &= \{W(L) : W\in W(OPT)\},\\
    W_R(OPT) &= \{W(R) : W\in W(OPT)\}.
  \end{align*}
  Now define
  \[
    \mathcal{W}' = W(OPT) \cup W_L(OPT) \cup W_R(OPT).
  \]
  Next, for every interval $I$, there exists exactly $\Delta-1$ many values of $j$ such that $I\subseteq W_j^{\Delta}$.
  Therefore $\vert\mathcal{W}'\vert\leq 3\cdot\Delta\cdot\vert OPT \vert$.
  Also, the set of windows marked as \textsf{active} must be a subset of $\mathcal{W}'$.
  If not, there is a window marked as \textsf{active} that intersects no interval in $OPT$, and so there is an interval in $\mathcal{S}$ that intersects no interval in $OPT$.
  This contradicts the maximality of a maximum independent set.
  It follows that the total space used by the algorithm is bounded by
  \[
    O(s)\cdot O(\Delta\cdot\vert OPT\vert) = O(s\cdot\Delta\cdot\vert OPT \vert).
  \]

  For any window $W$, the set of intervals in $\mathcal{S}$ contained in $W$ is independent of the order of $\mathcal{S}$.
  Therefore, the substream of $\mathcal{S}$ on $W$ is also a random-order stream\footnote{This is where our argument adds on to the similar argument by \cite{ddk23}.}.

  Next, for each $j\in[\Delta]$, define the set of disjoint windows $\mathbb{W}_j = \{W_{j + a\Delta}^{\Delta} : a\in\mathbb{Z}\}$ and define $OPT_j = \{I\in OPT : \exists W\in\mathbb{W}_j\text{ s.t. }I\subseteq W\}$.
  Then, using that every interval in $OPT$ is contained in exactly $\Delta-1$ windows and only a single window from any given $\mathbb{W}_j$,
  \[
    \sum_{j=1}^{\Delta} \vert OPT_j \vert = (\Delta-1)\vert OPT \vert.
  \]
  Therefore, there exists some $k\in[\Delta]$ such that
  \begin{align}
    \vert OPT_k \vert \geq \frac{\Delta-1}{\Delta}\cdot\vert OPT \vert.\label{eq: OPT_k}
  \end{align}
  We now consider the set of windows $\mathbb{W}_k$.
  Each window in $\mathbb{W}_k$ is disjoint so $\bigcup_{W\in\mathbb{W}_k} OUT(W)$ is a valid independent set of the intervals of $\mathcal{S}$.
  Then, taking the expectation over the random order of the stream,
  \begin{align*}
    \mathbb{E}[\vert OUT \vert] &\geq \mathbb{E}\left[\sum_{W\in\mathbb{W}_k} \vert OUT(W) \vert\right]\tag{$\bigcup_{W\in\mathbb{W}_k} OUT(W)$ is a valid independent set.}\\
                                &= \sum_{W\in\mathbb{W}_k} \mathbb{E}[\vert OUT(W) \vert]\tag{Linearity of Expectation and the substream of $\mathcal{S}$ on $W$ is also random-order.}\\
      &\geq \sum_{W\in\mathbb{W}_k} \alpha\cdot\vert \{I\in OPT : I\subseteq W\} \vert\tag{$\mathcal{A}$ computes an $\alpha$-approximation in expectation.}\\
      &= \alpha\cdot\sum_{W\in\mathbb{W}_k} \vert \{I\in OPT : I\subseteq W\} \vert\\
      &= \alpha\cdot\left\vert\bigcup_{W\in\mathbb{W}_k}\{I\in OPT : I\subseteq W\}\right\vert\tag{Each $W\in\mathbb{W}_k$ is disjoint, so each $I$ is contained in exactly one $W\in\mathbb{W}_k$.}\\
      &= \alpha\cdot\vert\{I\in OPT: \exists W\in\mathbb{W}_k\text{ s.t. }I\subseteq W\}\vert\\
      &= \alpha\cdot\vert OPT_k \vert\tag{Definition of $OPT_k$.}\\
      &\geq \alpha\cdot\frac{\Delta-1}{\Delta}\cdot\vert OPT \vert.\tag{Equation \ref{eq: OPT_k}.}
  \end{align*}

  This meta-algorithm for unrestricted domains is deterministic.
  Therefore, the final algorithm for unrestricted domains is deterministic if and only if $\mathcal{A}$ is deterministic.
\end{proof}

\subsection{Completing the Analysis}

We now combine Algorithm \ref{alg: alg} with Lemma \ref{lem: unrestricted domain} to complete the proof of Theorem \ref{thm: alg}.

\ALG*

\begin{proof}
    The proof is by combining Algorithm \ref{alg: alg} with Lemma \ref{lem: unrestricted domain} to obtain an algorithm for random-order unit-length interval streams without any restrictions on the input domain.
    We first consider the space used by this approach.
    We then lower bound the expected approximation factor achieved on unrestricted domains with a function that only depends on the integer window length $\Delta$ (from Lemma \ref{lem: unrestricted domain}), and finally numerically optimise this function.

    From Subsection \ref{subs: expected approx}, on a window of length $\Delta$, we run an algorithm with domain size $\Delta + 2$.
    Then, from Lemmas \ref{lem: space used} and \ref{lem: unrestricted domain}, this yields an algorithm for unrestricted domains that uses
    \begin{align*}
        O(S(\Delta + 2)\cdot\Delta\cdot\vert OPT \vert) &= O\left( 4^{\Delta + 2}\cdot(\Delta+2)^{\Delta+2}\cdot\Delta\cdot\vert OPT \vert \right)\\
        &= O\left( 4^\Delta \cdot (\Delta+2)^{\Delta+3} \cdot \vert OPT \vert \right),
    \end{align*}
    words of space, where $S(k)$ denotes the space used by an instance of the algorithm instantiated on a domain of length $k$ (from Lemma \ref{lem: space used}).
    Then, when $\Delta = O(1)$, this space collapses to become $O(\vert OPT \vert)$.
  
    We now consider the approximation factor of the unrestricted algorithm.
    Observe that a largest independent set in the  window $[0,\Delta)$ is of size at most  $\Delta - 1$.
    Therefore, combining Lemmas \ref{lem: alg restricted domains} and \ref{lem: unrestricted domain}, for any $\Delta$, there is an algorithm for random-order unit-length interval selection on unrestricted domains that achieves an expected approximation factor of
    \begin{align}
      \frac{\Delta - 1}{\Delta}\cdot\min_{\alpha\in\{1,\dots,\Delta-1\}}\left\{ \frac{1}{\alpha} + \frac{1}{\alpha^2}\cdot\sum_{i=1}^{\alpha} \max\left\{ \begin{array}{lr} out(i - 1) + out(\alpha - i -1)\\ out(\alpha - i) + out(i - 2)\end{array} \right\} \right\}\label{eq: overall approx}.
    \end{align}
  
    It remains to pick a suitable constant $\Delta$.
    Computing (\cref{eq: overall approx}) numerically with $\Delta = 5,000$ (using the recurrence relation for $out(\cdot)$ of \cref{lem: alg restricted domains}) yields
    a lower bound of (slightly) more than $0.7401$ for the expected approximation factor, completing the proof.
    We discuss this choice of $\Delta$ further in Appendix \ref{app: different gamma}.
\end{proof}

\section{The Lower Bound}\label{sec: LBs}

We now establish our lower bound result as stated in Theorem~\ref{thm: exp INDEX LB}.

Given a random order \textsf{Unit Interval Selection} algorithm $\mathcal{A}$, we define protocol $Q$ that uses $\mathcal{A}$ to solve the $\textsf{Index}_t$ problem:

\begin{tcolorbox}[standard jigsaw,opacityback=0,width=0.95\textwidth,breakable]
  \textbf{Protocol $Q$ for $\textsf{INDEX}_t$:}\label{protocol: Q}

  \paragraph*{Input:}
  Let Alice and Bob be two players holding a $\textsf{INDEX}_t$ instance $(X,A)$ from the uniform distribution $\mu$.
  Let $\mathcal{A}$ be any algorithm for random-order \textsf{Unit Interval Selection}.
  For each $i\in[t]$, define the uniformly and independently distributed binary random variable $Y_i\in\{0,1\}$.
  Furthermore, define the uniformly distributed random variables $B_L = B_R \in [t]$.
  Let $T$ denote the set containing these random variables, that is, $T = \{Y_1,\dots,Y_t,B_L,B_R\}$.

  \paragraph*{The Protocol:}
    \begin{enumerate}
    \item Alice and Bob use shared public randomness to sample each variable in $T$. They also sample a bijection $\sigma:T\rightarrow[\vert T\vert]$.

    \item Let $j' = \min \{\sigma(B_L), \sigma(B_R)\}$.

    \item For each $i\in[t]$ such that $\sigma(Y_i) < j'$, Alice constructs the interval
      \[
        I[i] = \left[\frac{i}{t+1} + \frac{X[i]}{t^2}, 1 + \frac{i}{t+1} + \frac{X[i]}{t^2}\right].
      \]

     \item Alice runs $\mathcal{A}$ on these constructed intervals using the ordering $\sigma$. Note that $\sigma$ defines an ordering on these intervals by associating each $I[i]$ with $Y[i]$.
      Alice then sends the state of $\mathcal{A}$ to Bob as the message.

     \item For each $i\in[t]$ such that $\sigma(Y[i]) > j'$, Bob constructs the interval
     \[
       I[i] = \left[\frac{i}{t+1} + \frac{Y[i]}{t^2}, 1 + \frac{i}{t+1} + \frac{Y[i]}{t^2}\right].
     \]
     Bob also constructs the intervals
     \[
        J_L = \left[ \frac{A}{t+1} - \frac{1}{t^3} - 1, \frac{A}{t+1} - \frac{1}{t^3} \right]
     \]
     and
     \[
        J_R = \left[ 1 + \frac{A}{t+1} + \frac{1}{t^2} + \frac{1}{t^3}, 2 + \frac{A}{t+1} + \frac{1}{t^2} + \frac{1}{t^3} \right].
     \]

    \item Bob uses the message from Alice to continue running $\mathcal{A}$ on these intervals using the ordering $\sigma$, associating $J_L$ with $B_L$, and $J_R$ with $B_R$, under $\sigma$.
    \end{enumerate}

    \paragraph*{Output:}
    If either $\cA$ produces an output of size at most $2$, or $\sigma(Y[A]) > j'$, then Bob outputs a uniform random bit.
    Otherwise, it must be that $\cA$ outputs a solution of size three, and $\sigma(Y[A]) < j'$.
    As observed below, this solution must contain the interval $I[A]$. Then, if $I[A]$ is located at the position $[\frac{A}{t+1}, 1 + \frac{A}{t+1}]$ then Bob outputs the bit $x=0$, otherwise the interval must be located at position $[\frac{A}{t+1} + \frac{1}{t^2}, 1 + \frac{A}{t+1} + \frac{1}{t^2}]$ and Bob outputs the bit $x=1$.
\end{tcolorbox}
Protocol $Q$ works as follows. Given a uniform random input $(X, A)$ to $\textsf{INDEX}_t$, Alice and Bob first create a random interval instance that is fully described by uniform random bits $Y_1, \dots, Y_t$ as well as the index $A$ from the $\textsf{INDEX}_t$ input: Each bit $Y_i$ corresponds to an interval $I[i]$ so that the intervals $I[1], \dots, I[t]$ all mutually overlap as in Figure~\ref{fig:intervals}. The role of the bit $Y_i$ is as follows: If $Y_i = 1$ then the interval $I[i]$ is shifted slightly to the right, while this does not happen if $Y_i = 0$. Then, the index $A$ translates into two additional intervals $J_L$ and $J_R$ that surround the interval $I[A]$ so that $\{J_L, J_R, I[A] \}$ constitutes the unique independent set of size $3$.

We observe that the instance created so far is a uniform random instance since all the variables $Y_1, \dots, Y_t, A$ are independent uniform random variables. 

Next, Alice and Bob choose a uniform random ordering $\sigma$ of the intervals in this instance. We stress that, by construction, $\sigma$ is independent of the intervals created, i.e., of the quantities $Y_1, \dots, Y_t$ as well as $A$. 

Alice now replaces some of the bits $Y_i$ by the values $X_i$ as follows: If in the uniform random ordering $\sigma$ the interval $I[i]$ arrives before both $J_L$ and $J_R$, then $Y_i$ is updated to the value $X_i$. Distributionally speaking, we observe that this process replaces uniform random bits $Y_i$ by different uniform random bits $X_i$ (recall that the input to the $\textsf{INDEX}_t$ instance is a uniform random instance). Hence, the distribution of the resulting interval instance is a uniform distribution, and, most importantly, the order of the intervals is random and independent of the instance. The reduction thus creates a random order stream.

We observe that Alice only updates intervals that arrive before the intervals $J_L$ and $J_R$, or, in other words, the decision as to which intervals are updated depends on the stream order. One may falsely conclude that the resulting instance may therefore be correlated with the arrival order $\sigma$, which would mean that the resulting stream is not in uniform random order. We stress, however, that this is not the case since this process only replaces uniform random bits by a different set of uniform random bits. It is therefore not possible to establish a correlation.

Next, we see that whenever Alice encoded the correct bit $X[A]$ into interval $I[A]$, which, as we will see, happens with probability $1/3$, and if, at the same time, the algorithm $\mathcal{A}$ succeeds, then Alice and Bob are able to solve the underlying $\textsf{INDEX}_t$ instance. Otherwise, Bob outputs a uniform random bit. We will see that if $\mathcal{A}$ succeeds with probability greater than $2/3$ then the protocol $Q$ has a success probability strictly larger than $1/2$. The lower bound for $\textsf{INDEX}_t$ as stated in Theorem~\ref{thm: INDEX hardness} then implies that $Q$ must communicate $\Omega(t)$ bits, and thus, algorithm $\mathcal{A}$ must also use at least as much space.

\begin{figure}[H]
  \centering
  \begin{tikzpicture}[scale=0.7]
    \foreach \x in{1}
    {
    \draw[very thick] (\x,0.5) -- (3+\x,0.5);
    \draw[very thick] (\x+3/7,0.8) -- (3+\x+3/7,0.8);
    \draw[very thick] (\x+6/7,1.1) -- (3+\x+6/7,1.1);
    \draw[very thick] (\x+9/7,1.4) -- (3+\x+9/7,1.4);
    \draw[very thick] (\x+12/7,1.7) -- (3+\x+12/7,1.7);
    \draw[very thick] (\x+15/7,2) -- (3+\x+15/7,2);
    \draw[very thick] (\x+18/7,2.3) -- (3+\x+18/7,2.3);
    }
    \node at (5,2.7) {$I[0],\dots,I[t-1]$};

    \draw[very thick, red] (-3/7,1.7) -- (1+11/7,1.7);
    \node[red] at (15/14,2) {$J_L$};
    \draw[very thick, red] (1+3+13/7,1.7) -- (1+3+13/7+3,1.7);
    \node[red] at (1+3+13/7+1.5,2) {$J_R$};
\end{tikzpicture}
  \caption{An example set of intervals constructed by Protocol $Q$. The black intervals correspond to each $I[i]$, and the red intervals show $J_L$ and $J_R$. Whether or not each $I[i]$ interval was constructed using $X[j]$ or a public random bit depends on the sampled permutation $\sigma$. \label{fig:intervals}}
\end{figure}
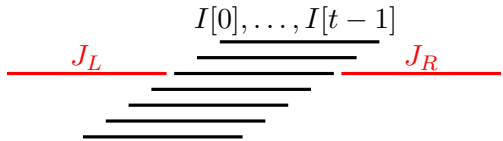

\INDEXLB*

\begin{proof}
    Let $\mathcal{A}$ be a randomized algorithm for \textsf{Unit Interval Selection} with approximation factor $2/3+\delta$, for any small constant $\delta$, that succeeds with probability at least $2/3 + \epsilon$, for some $\epsilon > 0$, on every input instance, where the probability is taken over both the stream order and the random coin flips of the algorithm. We will first prove that $\mathcal{A}$ must use $\Omega(n)$ space, i.e., we establish result 2. Then, we argue that result 2 immediately implies result 1.

    To see that result 2 holds, we first observe that, whenever algorithm $\mathcal{A}$ succeeds then it must output a solution of size $3$ since its approximation factor is strictly larger than $2/3$ and every instance created contains an optimal solution of size $3$. Second, as argued above, no matter the input instance created, it is presented to $\mathcal{A}$ in uniform random order, which implies that the success probability of the algorithm is indeed $2/3+\epsilon$. Third, we see that, with probability $1/3$, the interval $I[A]$ is created by Alice and thus correctly reflects the bit $X[A]$. We are thus interested in the probability that both $I[A]$ reflects $X[A]$ and the algorithm $\mathcal{A}$ succeeds on the resulting instance. 

    Using the events $\mathcal{E}_1 = $ "Alice creates $I[A]$" and $\mathcal{E}_2 = $"$\mathcal{A}$ succeeds", we obtain:
    \begin{align}
        \Pr[\mathcal{E}_1 \text{ and } \mathcal{E}_2]  
         = \Pr[\mathcal{E}_1] \cdot \Pr[\mathcal{E}_2 \ | \ \mathcal{E}_1 ] 
         = \frac{1}{3}  \cdot \Pr[\mathcal{E}_2 \ | \ \mathcal{E}_1 ]\ . \label{eqn:123}
    \end{align}
    Next, we use the fact that the success probability of $\mathcal{A}$ is at least $2/3+\epsilon$. Hence,
    \begin{align*}
        2/3 + \epsilon & \le \Pr[\mathcal{E}_2] = \Pr[{\mathcal{E}_1}] \cdot \Pr[\mathcal{E}_2 \ | \ \mathcal{E}_1] + (1- \Pr[{\mathcal{E}_1}]) \cdot \Pr[\mathcal{E}_2 \ | \ \overline{\mathcal{E}_1}] \\
        & = \frac{1}{3} \cdot \Pr[\mathcal{E}_2 \ | \ \mathcal{E}_1] + \frac{2}{3} \cdot \Pr[\mathcal{E}_2 \ | \ \overline{\mathcal{E}_1}] \\
        & \le  \frac{1}{3} \cdot \Pr[\mathcal{E}_2 \ | \ \mathcal{E}_1] + \frac{2}{3} \cdot 1 \ ,
    \end{align*}
    which implies $
        \epsilon \le \Pr[\mathcal{E}_2 \ | \ \mathcal{E}_1] \ . $
Using this bound in Equality~\ref{eqn:123} yields $\Pr[\mathcal{E}_1 \text{ and } \mathcal{E}_2] \ge \frac{1}{3} \cdot \epsilon$. In other words, with probability $\frac{1}{3} \cdot \epsilon$, Bob solves the underlying $\textsf{Index}_t$ instance. 

Lastly, in all other cases, i.e., with probability $1-\epsilon / 3$, Bob outputs a uniform random bit. Hence, the protocol $Q$ solves the $\textsf{Index}_t$ instance with probability $\epsilon / 3 + (1 - \epsilon / 3) \cdot \frac{1}{2} = \frac{1}{2} + \epsilon / 6$. Then, by Theorem~\ref{thm: INDEX hardness}, the resulting protocol must use a message of size $\Omega(t)$ bits, and since the message size is identical to the space used by our algorithm, the same lower bound applies to the space usage of $\mathcal{A}$. Last, since $n = t+2$ in our construction, result 2 follows.

To see that result 1 holds, consider an algorithm $\mathcal{A}$ for \textsf{Unit Interval Selection} with expected approximation factor $\frac{8}{9} + \epsilon$, for some $\epsilon > 0$, where the expectation is taken over the random coin flips of the algorithm as well as the stream order. We will now argue that $\mathcal{A}$ achieves an approximation factor of at least $2/3 + \epsilon$ with probability $2/3 + \epsilon$. Then, by result 2 we obtain that $\mathcal{A}$ must use $\Omega(n)$ bits of space, thus establishing result 1.

We compute:
\begin{align*}
    \frac{8}{9} + \epsilon & \le \Exp[\text{approx factor}] \\
    & \le \Pr[\text{approx factor} \le 2/3 +\epsilon] \cdot (2/3 + \epsilon) + \Pr[\text{approx factor} > 2/3  + \epsilon] \cdot 1 \\
    & = (1 - \Pr[\text{approx factor} > 2/3 + \epsilon ]) \cdot (2/3 + \epsilon)+ \Pr[\text{approx factor} > 2/3 +\epsilon] \ , 
\end{align*}
which implies $\Pr[\text{approx factor} > 2/3 + \epsilon ] \ge \frac{2}{3} + \epsilon$ as desired.
\end{proof}

\section{Conclusion} \label{sec: conclusion}
In this paper, we gave a one-pass random order streaming algorithm for \textsf{Unit Interval Selection} with expected approximation factor $0.7401$ that uses space $O(|OPT|)$, where the expectation is taken over the order of the stream. We also proved that, within this space constraint, no algorithm can achieve an expected approximation factor beyond $8/9$, and no algorithm can have an approximation factor better than $2/3$ with probability above $2/3$. 

We conclude with two open questions. 
\begin{enumerate}
\item Our results are not tight. Is there either an algorithm with improved expected approximation guarantee or is it possible to prove a stronger lower bound? 

\item We have not addressed arbitrary-length intervals. Is there an algorithm with expected approximation factor above $1/2$ for \textsf{Interval Selection} on arbitrary-length intervals?
\end{enumerate}

\bibliographystyle{plainurl}
\bibliography{random-order-intervals}

\newpage

\appendix

\section{The Expected Approximation Factors Achieved for Different Values of \texorpdfstring{$\Delta$}{Delta}}\label{app: different gamma}

As discussed in the proof of Theorem \ref{thm: alg}, the approximation factor achieved by our algorithm depends on the choice of $\Delta$.
For completeness, we now discuss the approximation factors achieved by different (constant) values of $\Delta$.

We do this by computing them numerically, and this can be done efficiently via dynamic programming.
Doing this shows us that the approximation factor appears to be an increasing function of $\Delta$, and, by choosing $\Delta = 5,000$, we achieve an expected approximation factor larger than $\sim0.7401$.
Interestingly, any integer $\Delta$ at least five suffices to beat the adversarial-order $2/3$-barrier.
This is also useful for giving us a way to 'manually prove' a stronger-than $2/3$-result by considering the simpler case of $\Delta = 5$.

A natural follow up question is then how much this can be improved by choosing an even larger value of $\Delta$.
To answer this, we can first decompose the expected approximation factor of the algorithm (Equation \ref{eq: overall approx}) into two components.
The term $(\Delta-1)/\Delta$ is the 'loss' accrued by going from a restricted domain to an unrestricted domain via Lemma \ref{lem: unrestricted domain}, and the remaining term is the approximation factor achieved within a restricted domain via Algorithm \ref{alg: alg}.
It is clear that $(\Delta-1)/\Delta \leq 1$, and so the entire quantity can be upper bounded by the approximation factor achieved within a restricted domain via Algorithm \ref{alg: alg}.

We can now compute these values numerically, and doing so gives us a bound of $\sim0.7401803$ with a window length of $\Delta = 100,000$.
This suggests that we could not hope to achieve too much more by simply picking a larger value of $\Delta$.

These ideas are illustrated in Figures \ref{fig: gamma leq 20} and \ref{fig: gamma leq 100}.
In both figures, for different values of $\Delta$, the blue dots represent the expected approximation factor achieved by the general algorithm for unrestricted domains (Theorem \ref{thm: alg}), and the orange dots represent the approximation factor achieved by Algorithm \ref{alg: alg} on restricted domains (Lemma \ref{lem: alg restricted domains}).
As discussed, the orange dots act as an upper bound for the blue dots.
The dotted red lines indicate the $2/3$-barrier from the adversarial order settings, and the dotted blue line indicates the value $0.740175$, which lies in between our computed values for both approximation factors.
Note that the dotted blue line is shown to aid intuition, but does not correspond to any theoretically computed values.
The dotted green line indicates the $8/9$-approximation barrier for algorithms that use space $o(n)$ (from Theorem \ref{thm: exp INDEX LB}).

\begin{figure}
    \centering
    \begin{minipage}{0.48\textwidth}
        \centering
        \includegraphics[width=1\textwidth]{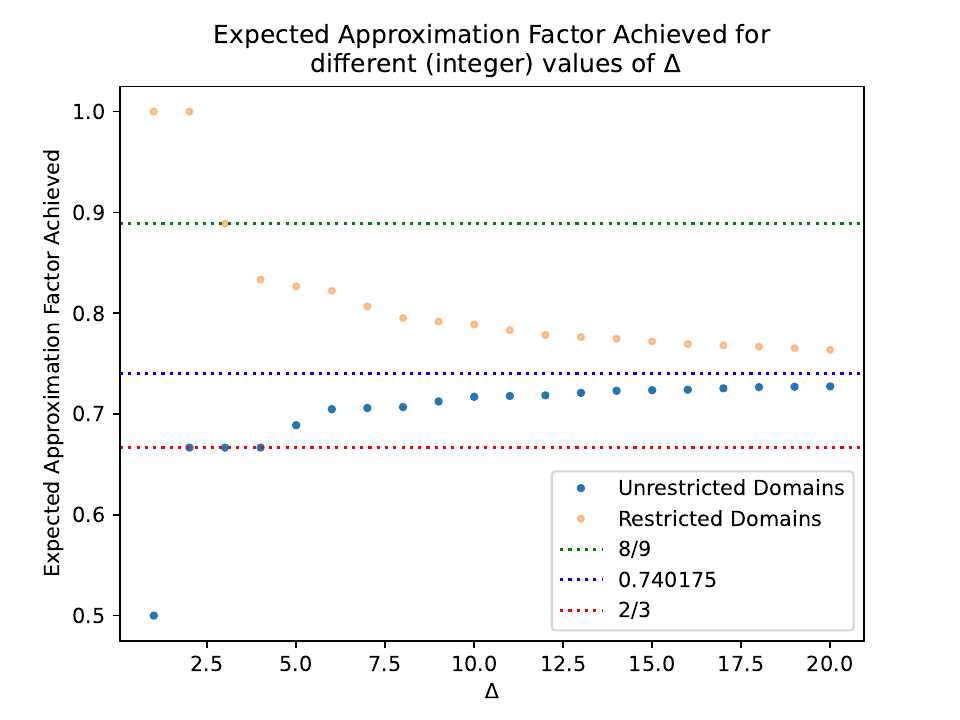} 
        \caption{The expected approximation factor achieved for each integer $\Delta \leq 20$.}
        \label{fig: gamma leq 20}
    \end{minipage}\hfill
    \begin{minipage}{0.48\textwidth}
        \centering
        \includegraphics[width=1\textwidth]{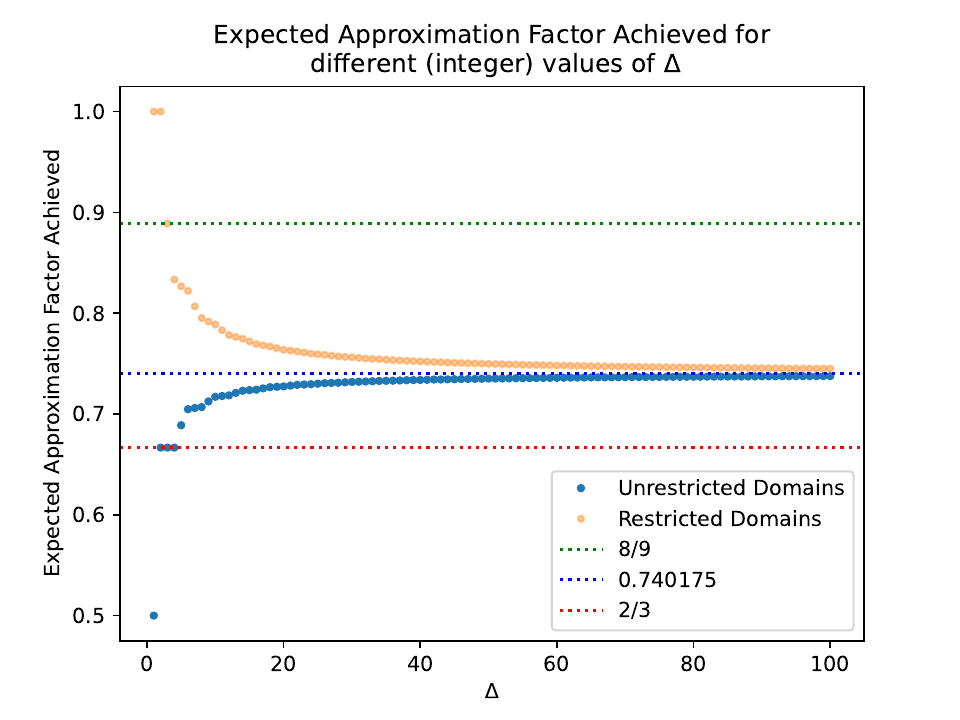} 
        \caption{The expected approximation factor achieved for each integer $\Delta \leq 100$.}
        \label{fig: gamma leq 100}
    \end{minipage}
\end{figure}

\end{document}